\documentclass[11pt]{article}

\usepackage{amsmath,amssymb,amsthm,color} 
\usepackage{empheq}
\usepackage[most]{tcolorbox}
\usepackage[makeroom]{cancel}
\usepackage{todonotes}

\newcommand{\ket}[1]{\big| #1 \big\rangle}  
\newcommand{\bra}[1]{\big\langle #1 \big|}  
\newcommand{\braket}[2]{\big\langle #1 \big| #2 \big\rangle}                 

\newcommand{\hypergeometricseries}[5]{{}_{#1}F_{#2}\!\left(\left.\!\!\!\begin{array}{c} #3 \\ #4 \end{array}\!\!\right| #5 \right)}

\newtheorem{prop}{Proposition}
\newtheorem{corollary}[prop]{Corollary}
\newtheorem{claim}[prop]{Claim}

\usepackage{soul}

\usepackage{authblk}
\begin{document}

\title{Multimarked Spatial Search by \\ Continuous-Time Quantum Walk
}


\author[1]{Pedro H. G. Lug\~ao}
\author[1]{Renato Portugal}
\author[2]{Mohamed Sabri}
\author[3]{Hajime Tanaka}

\affil[1]{\small National Laboratory of Scientific Computing (LNCC)

          Petr\'opolis, RJ, 25651-075, Brazil\vspace{5pt}
}

\affil[2]{\small Department of Mathematics

            The Open University of Sri Lanka
            
            Nawala, Nugegoda, 10250, Sri Lanka\vspace{5pt}
              }

\affil[3]{\small Research Center for Pure and Applied Mathematics

          Graduate School of Information Sciences
          
          Tohoku University, Sendai 980-8579, Japan 
              }


\maketitle

\begin{abstract}
The quantum-walk-based spatial search problem aims to find a marked vertex using a quantum walk on a graph with marked vertices. We describe a framework for determining the computational complexity of spatial search by continuous-time quantum walk on arbitrary graphs by providing a recipe for finding the optimal running time and the success probability of the algorithm. The quantum walk is driven by a Hamiltonian derived from the adjacency matrix of the graph modified by the presence of the marked vertices. The success of our framework depends on the knowledge of the eigenvalues and eigenvectors of the adjacency matrix. The spectrum of the Hamiltonian is subsequently obtained from the roots of the determinant of a real symmetric matrix $M$, the dimensions of which depend on the number of marked vertices. The eigenvectors are determined from a basis of the kernel of $M$. We show each step of the framework by solving the spatial searching problem on the Johnson graphs with a fixed diameter and with two marked vertices. Our calculations show that the optimal running time is $O(\sqrt{N})$ with an asymptotic probability of $1+o(1)$, where $N$ is the number of vertices. 

\

\end{abstract}

\section{Introduction}

Farhi and Gutmann~\cite{FG98} introduced the continuous-time quantum walk as a discrete-space version of the Schr\"odinger equation, where the positions of the particle are the vertices of a graph, specifically, a tree whose nodes represent solutions to a decision problem~\cite{sip12}. Childs and Goldstone~\cite{CG_04} used the continuous-time quantum walk to address the spatial search problem with one marked vertex, the goal being to locate the marked vertex when the walk departs from an initial state that is easy to prepare. The evolution is based on a Hamiltonian, which includes an extra term that depends on the location of the marked vertex and differs from the Hamiltonian used by Farhi and Gutmann, which is solely based on the graph's adjacency matrix. Childs and Goldstone analyzed the search on lattices, hypercubes, and complete graphs to determine whether the quantum-walk-based search algorithm is faster than the random-walk-based one on the same graphs. They obtained negative results in some cases, for instance, on $d$-dimensional lattices with $d\leqslant 4$. However, this should not be interpreted as definitive evidence that it is impossible to find a quantum-walk-based search algorithm that provides a speedup over the best random-walk-based case. In fact, using alternative discrete-time quantum walk models, it is possible to achieve a quadratic speedup for $d$-dimensional lattices with $d\geqslant 2$~\cite{AKR05,PF17}.

The origin of the quantum-walk-based spatial search problem lies in the confluence of two proposals: Farhi and Gutmann's 1998 suggestion for accelerating random walks on decision trees using a continuous-time quantum walk~\cite{FG98}, and Benioff's 2000 proposition to apply Grover's algorithm~\cite{Gro97} with the aim of speeding up classical search algorithms for a marked vertex on the two-dimensional lattice~\cite{Ben02}. Following the contributions of these authors, the area of quantum-walk-based spatial search algorithms has diverged into two different formulations in the literature. Both are significant as they expose distinct computational and mathematical aspects of quantum walks.

The first formulation investigates whether a quantum walk on a graph $G$ requires fewer steps to find a marked vertex than a classical random walk on the same graph $G$. This formulation was introduced by Shenvi et al.~\cite{SKW03}, who showed that a quantum walk on a hypercube locates a marked vertex faster than a random walk. In a context close related to the spatial search algorithm, Childs et al.~\cite{CFG02} showed that the propagation of a particle between a specific pair of nodes is exponentially faster when driven by a continuous-time quantum walk as compared to a random walk. 
Numerous papers in the literature address this formulation in both the discrete-time~\cite{SKW03,AP18,CPBS18,BLP21} and continuous-time~\cite{CG_04,CNAO16,PTB16,Won16,CRPM18,LCRLL19,HH19,WWW19,CNR20b,PBO21} cases. The results of this work align with this formulation. Experimental implementations of search algorithms by continuous-time quantum walk are described in~\cite{DGPSW20,WSXWJX20,BTPC21,QMWXWX22}.

The second formulation asks whether the number of steps that a classical random walk on a graph $G$ takes to find a marked vertex $v\in V(G)$ can be reduced by implementing a quantum walk on any graph $G'$, so that the quantum walk on $G'$ would uncover the necessary information to find the vertex $v$ in the original graph $G$. This formulation was introduced by Szegedy~\cite{Sze04a}, who showed that the quantum hitting time of a quantum walk on the bipartite graph $G'$ obtained from a graph $G$, where the random walk takes place, is quadratically less than the classical hitting time. However, to locate a marked vertex, it remains essential to prove that the success probability is $\Omega(1)$. Many papers in the literature have addressed the second formulation in both the discrete-time~\cite{MNRS11,MNRS12,KMOR16,AGJK20} and continuous-time~\cite{CNR20a,ACNR21} cases. In all these papers, the quantum walk takes place on a bipartite graph $G'$ derived from $G$ via a duplication process and the random walk takes place on an almost arbitrary graph $G$ with multiple marked vertices. Ambainis et al.~\cite{AGJK20} effectively resolved Szegedy's problem by showing that there is a quadratic speedup for finding a marked vertex with success probability $\tilde\Omega(1)$ on an arbitrary graph $G$ by discrete-time quantum walk on $G'$. Apers et al.~\cite{ACNR21} go along the same line by continuous-time quantum walk.
Somma and Ortiz~\cite{Somma2010} have described an interesting method to address combinatorial optimization problems that was used in~\cite{CNR20a} to modify the standard method proposed by Childs and Goldstone~\cite{CG_04} for the case with one marked vertex.

In this work, we address the first formulation of the spatial search problem by continuous-time quantum walk on arbitrary graphs with multiple marked vertices. Additionally, we provide a framework to determine the computational complexity of the search algorithm that is successful when the gap of two key eigenvalues $\lambda^\pm$ of the Hamiltonian is small enough. Note that our aim differs from the one described in references~\cite{CNAO16,CNR20b}, which focus on obtaining optimality conditions, namely, a time complexity that has a quadratic speedup compared to random walks. Our method may result in non-optimal time complexity, but it remains interesting if it is better than $O(N)$. For instance, the $d$-dimensional lattice with multiple marked vertices, where $d=2$ and $d=3$, are interesting candidates. These were analyzed in the reference~\cite{CG_04}, but only for the case with one marked vertex.
As an example of our framework, we analyze in full detail the search algorithm on the Johnson graph with two marked vertices. The framework is based on a real symmetric matrix $M$, which depends on the entries of the adjacency matrix and on the locations of the marked vertices. The fact that the determinant of this matrix is zero allows us to find the spectrum of the Hamiltonian. Then, we can find the gap between $\lambda^\pm$ and all quantities necessary to determine the algorithm's optimal running time and the success probability assuming that the gap tends to zero when the number of vertices increases.  The relevant quantities are basically the overlaps between the marked vertices and the initial condition with the eigenvectors of the Hamiltonian associated with $\lambda^\pm$. 

The Johnson graph $J(n,k)$ plays an important role in quantum walks because it is used in the element distinctness algorithm~\cite{MNRS11} and in many spatial search algorithms~\cite{Wong2016JPA,TSP22}. The vertices of $J(n,k)$ are $k$-subsets of a set with $n$ elements and two vertices are adjacent if and only if the intersection of these vertices is a $(k-1)$-subset. In this work, we use the Johnson graph $J(n,k)$ with two marked vertices as an example of all steps of our framework for determining the time complexity of the spatial search algorithm by continuous-time quantum walk. We show that the optimal running time to find a marked vertex is $\pi k\sqrt{N}/(2\sqrt{2})$ with asymptotic probability $1+o(1)$ when $k$ is fixed. 

The structure of this paper is as follows. In Sec.~\ref{sec:frame}, we lay out a framework for determining the computational complexity of the spatial search algorithm by continuous-time quantum walk on graphs with multiple marked vertices. In Sec.~\ref{sec:JohnsonGraph}, we apply this framework to analyze with mathematical rigor the spatial search algorithm on the Johnson graph with two marked vertices. Finally, in Sec.~\ref{sec:conc}, we present our concluding remarks.

\section{Multimarked framework}\label{sec:frame}

Let $G(V,E)$ be a finite connected simple graph with vertex set $V$ and edge set $E$.
We associate $G$ with a Hilbert space $\mathcal{H}$ with computational basis $\{|v\rangle : v\in V\}$, as is usually done in the definition of the continuous-time quantum walk~\cite{FG98}.
Let 
$W$ be the set of marked vertices.
We consider the time-independent Hamiltonian of the form~\cite{CG_04}
\begin{equation}\label{eq:hamiltonian}
	H=-\gamma A- \sum_{w\in W} \ket{w} \bra{w},
\end{equation}
where $A$ denotes the adjacency matrix of $G(V,E)$, and $\gamma$ is a real and positive parameter.
The spatial search algorithm starts with the initial state $\ket{\psi(0)}$, which is a normalized eigenvector of $-\gamma A$ associated with the largest eigenvalue of $A$. If $G$ is regular,  $\ket{\psi(0)}$ is the uniform superposition of the computational basis
\begin{equation*}
	\ket{\psi(0)}=\frac{1}{\sqrt{N}}\sum_{v\in V}|v\rangle,
\end{equation*}
where $N$ is the number of vertices.  
The quantum state at time $t$ is therefore given by
\begin{equation*}
	\ket{\psi(t)}=\mathrm{e}^{-\mathrm{i}Ht}\ket{\psi(0)}.
\end{equation*}
The probability of finding a marked vertex at time $t$ is
\begin{equation}\label{eq:p(t)}
	p(t) = \sum_{w\in W}\left|\braket{w}{\psi(t)}\right|^2.
\end{equation}
Our goal is to find the optimal values of parameters $t$ and $\gamma$ so that the success probability is as high as possible.

\subsection{Eigenvalues and eigenvectors of the Hamiltonian}\label{subssec:spectrum}

For an operator $U$, let $\sigma(U)$ denote the spectrum of $U$.
Suppose that the adjacency matrix $A$ of $G(V,E)$ has exactly $k+1$ distinct eigenvalues $\phi_0>\phi_1>\dots>\phi_k$.
For $0\leqslant \ell\leqslant k$, let $P_{\ell}$ denote the orthogonal projector onto the eigenspace of $A$ in $\mathcal{H}$
for the eigenvalue $\phi_{\ell}$, that is,
\begin{equation*}
	A = \sum_{\ell=0}^k \phi_{\ell} P_{\ell}.
\end{equation*}
Let $\lambda$ and $\ket{\lambda}$ be an eigenvalue and a normalized eigenvector of $H$, respectively, that is
\begin{equation*}
	H\ket{\lambda}=\lambda \ket{\lambda}
\end{equation*}
and $\braket{\lambda}{\lambda}=1$. Hamiltonian $H$ and operator $-\gamma A$ may share common eigenvalues and eigenvectors, as shown in the next Proposition.

\begin{prop}\label{prop:lambda_in_sigma_A}
Let $| \lambda \rangle $ be an eigenvector of $H$ associated with eigenvalue $\lambda $. Then,
$\lambda\in \sigma(-\gamma A)$ and $(-\gamma A)\ket{\lambda}=\lambda\ket{\lambda}$ if and only if $\braket{w}{\lambda}=0$ for all $w\in W$.
\end{prop}
\begin{proof}
Using the definition of $H$, we obtain
\begin{equation*}
	 \sum_{w\in W} \ket{w} \braket{w}{\lambda}=-(\gamma A+H)\ket{\lambda}.
\end{equation*}
If $\lambda\in \sigma(-\gamma A)$ and $(-\gamma A)\ket{\lambda}=\lambda\ket{\lambda}$, the right-hand side is zero and the entries of the vector $\sum_{w\in W} \ket{w} \braket{w}{\lambda}$ must be zero.
Then, we have $\braket{w}{\lambda}=0$ for all $w\in W$.
On the other hand, if $\braket{w}{\lambda}=0$ for all $w\in W$, the left-hand side is zero. Since $H\ket{\lambda}=\lambda\ket{\lambda}$, we have $\lambda\in \sigma(-\gamma A)$ and $(-\gamma A)\ket{\lambda}=\lambda\ket{\lambda}$.
\end{proof}

In Eq.~\eqref{eq:p(t)}, it is evident that eigenvectors $\ket{\lambda}$, which satisfy $\braket{w}{\lambda}=0$ for all $w \in W$, do not contribute to the computation of the probability $p(t)$ for locating a marked vertex. These eigenvectors correspond to eigenvalues $\lambda$ that are part of the spectrum of $-\gamma A$. In the continuation, we are interested in eigenvectors $\ket{\lambda}$ for which $\braket{w}{\lambda} \neq 0$, at least for one marked vertex $w$. The eigenvectors $\ket{\lambda}$ fulfilling the condition $\braket{\psi(0)}{\lambda}= 0$ are also irrelevant to the dynamics. To discern which eigenvectors of $H$ should be disregarded, the following proposition is necessary.
\begin{prop}\label{prop:detM=0}
    Let $\lambda $ be a real number such that  $\lambda\notin\sigma(-\gamma A)$.
    Define the real symmetric matrix $M^{\lambda}=\big(M_{ww'}^{\lambda}\big)_{w,w'\in W}$, where
\begin{equation}\label{eq:Mww}
	M_{ww'}^{\lambda}=\delta_{ww'}+ \sum_{\ell=0}^k \frac{\bra{w}P_{\ell}\ket{w'}}{\lambda+\gamma\phi_{\ell}}.
\end{equation}
Then, $\lambda$ is an eigenvalue of $H$ if and only if $\det(M^{\lambda})=0$.
\end{prop}

\begin{proof}
($\Rightarrow$)
Assume that $\lambda$ is an eigenvalue of $H$ and $\ket{\lambda}$ is the corresponding eigenvector. From the definition of $H$, we have
\begin{equation}\label{before master equation}
    P_{\ell} H \ket{\lambda} = -P_{\ell} \left(\gamma A + \sum_{w \in W} \ket{w}\bra{w}\right) \ket{\lambda},
\end{equation}
and since $\lambda \notin \sigma(-\gamma A)$, we can rewrite the above equation as
\begin{equation}\label{masterequation}
    P_{\ell} \ket{\lambda} = -\frac{1}{\lambda + \gamma\phi_{\ell}} \sum_{w \in W} \braket{w}{\lambda} P_{\ell} \ket{w}.
\end{equation}
Given that $\sum_{\ell=0}^k P_{\ell} = I$, for a fixed marked vertex $w$, we obtain
$\braket{w}{\lambda} = \sum_{\ell=0}^k \bra{w} P_{\ell} \ket{\lambda}.$
Then, using Eq.~\eqref{masterequation}, we derive
\begin{equation*}
    \braket{w}{\lambda} = -\sum_{\ell=0}^k \frac{1}{\lambda + \gamma\phi_{\ell}} \left(\sum_{w' \in W} \braket{w'}{\lambda} \bra{w} P_{\ell} \ket{w'}\right),
\end{equation*}
which can be rewritten as
\begin{equation}\label{0-eigenvector}
    \sum_{w' \in W} M_{ww'}^{\lambda} \braket{w'}{\lambda} = 0,
\end{equation}
where $\big(M_{ww'}^{\lambda}\big)_{w, w' \in W}$ is as defined in Eq.~\eqref{eq:Mww}.
Since $\lambda\not\in\sigma(-\gamma A)$, the vector $\braket{w}{\lambda} \big\rvert_{w \in W}$ is nonzero by Proposition~\ref{prop:lambda_in_sigma_A}, and hence it is a $0$-eigenvector of $M^{\lambda}$.
It follows that $\det\big(M^{\lambda}\big) = 0.$

\

($\Leftarrow$)
Assume $\lambda$ is a real number not in $\sigma(-\gamma A)$ such that $\det(M^\lambda) = 0$.
$M^\lambda$ has a 0-eigenvector, denoted as $u(w) \big\rvert_{w \in W}$. 
Construct $k+1$ vectors $\ket{\lambda_\ell}$, inspired by Eq.~\eqref{masterequation}, as follows:
\begin{equation*}
    \ket{\lambda_\ell} = -\frac{1}{\lambda + \gamma\phi_\ell} \sum_{w \in W} u(w) P_\ell \ket{w}.\label{plLambd}
\end{equation*}
Note that $P_{\ell}\ket{\lambda_{\ell'}}=\delta_{\ell \ell'}\ket{\lambda_\ell}$, or equivalently, $A\ket{\lambda_\ell}=\phi_{\ell}\ket{\lambda_\ell}$.
Our candidate for an eigenvector of $H$ associated with $\lambda$ is
\begin{equation}\label{eq:lambda_candidate}
    \ket{\lambda} = \sum_{\ell=0}^k \ket{\lambda_\ell}.
\end{equation}
First, we claim that $\ket{\lambda}$ is nonzero.
Indeed, if $\ket{\lambda}=0$, then we have $0=P_{\ell}\ket{\lambda}=\ket{\lambda_{\ell}}$ for all $\ell$, and hence
\begin{equation*}
    0=-\sum_{\ell=0}^k (\lambda+\gamma\phi_{\ell})\ket{\lambda_{\ell}} = \sum_{\ell=0}^k \sum_{w \in W} u(w) P_\ell \ket{w} = \sum_{w \in W} u(w) \ket{w}
\end{equation*}
by $\sum_{\ell=0}^k P_{\ell} = I$, but this is impossible because $u(w) \big\rvert_{w \in W}$ is nonzero.
We next claim that $H\ket{\lambda} = \lambda\ket{\lambda}$. Indeed, using the definitions of $H$ and $\ket{\lambda_\ell}$, we have
\begin{align*}
    H\ket{\lambda} &= \sum_{\ell=0}^k \frac{\gamma\phi_\ell}{\lambda + \gamma \phi_\ell} \sum_{w \in W} u(w) P_\ell \ket{w} + \\
    &\quad \sum_{w, w' \in W} \left(\sum_{\ell=0}^k \frac{\bra{w} P_\ell \ket{w'}}{\lambda + \gamma\phi_\ell}\right) u(w') \ket{w}.
\end{align*}
The expression in parentheses becomes $M_{ww'}^{\lambda} - \delta_{ww'}$ by Eq.~\eqref{eq:Mww}.
Since $u(w) \big\rvert_{w \in W}$ is a $0$-eigenvector of $M^{\lambda}$,
we find
\begin{align*}
    H\ket{\lambda} = \sum_{\ell=0}^k \frac{\gamma\phi_\ell}{\lambda + \gamma \phi_\ell} \sum_{w \in W} u(w) P_\ell \ket{w} - 
     \sum_{w \in W} u(w) \ket{w}.
\end{align*}
Applying $\sum_{\ell=0}^k P_\ell$ on the second term of the right-hand side, the equation simplifies to
\begin{align*}
    H\ket{\lambda} &= \lambda \left(-\sum_{\ell=0}^k \frac{1}{\lambda + \gamma \phi_\ell} \sum_{w \in W} u(w) P_\ell \ket{w}\right),
\end{align*}
which concludes the proof, considering Eq.~\eqref{eq:lambda_candidate}.
\end{proof}

Given a graph \(G(V,E)\) with a set \(W\) of marked vertices, the proof of Proposition~\ref{prop:detM=0} outlines a method for determining the eigenvalues \(\lambda \notin \sigma(-\gamma A)\) of the Hamiltonian~\eqref{eq:hamiltonian}, assuming we know the eigenvalues and eigenprojections of the adjacency matrix. This method involves solving the equation \(\det(M^{\lambda}) = 0\) for an unknown real number \(\lambda\). Note that the numerator of \(\det(M^{\lambda})\) is a polynomial \(P(\lambda)\) of degree at most \((k+1)|W|\) in the indeterminate \(\lambda\). There can be at most \((k+1)|W|\) roots of \(P(\lambda) = 0\), some of which may lie in \(\sigma(-\gamma A)\). Additionally, since \(A\) is a real symmetric matrix, \(M^{\lambda}\) is also real symmetric and thus orthogonally diagonalizable.

The proof of Proposition~\ref{prop:detM=0} also shows how to identify a basis for the eigenspace associated with $\lambda$. The multiplicity of $\lambda$ is equal to the multiplicity of the zero eigenvalue of $M^\lambda$. Given an orthogonal basis of $0$-eigenvectors of $M^\lambda$, and applying Eq.~\eqref{eq:lambda_candidate}, we can derive an orthogonal basis for the $\lambda$-eigenspace of $H$. Consequently, as a corollary, the multiplicity of $\lambda$ is constrained by the upper bound of $|W|$.

Given that the initial state is $\ket{\phi_0}$, an eigenvector $\ket{\lambda}$ such that $\braket{\lambda}{\phi_0}= 0$ does not affect the quantum walk dynamics. The criterion to ensure that $\braket{\lambda}{\phi_0}\neq 0$ when eigenvector $\ket{\lambda}$ is obtained using Eq.~\eqref{eq:lambda_candidate} is
\begin{equation}\label{eq:cond_psi0}
\sum_{w\in W} u(w)\braket{\phi_0}{w}\neq 0.
\end{equation}
In the case of a regular graph $G(V,E)$, it is necessary to assert that $\sum_{w\in W} u(w)\neq 0$. 

\begin{prop}\label{from PF}
Let $\lambda^-$ be the smallest eigenvalue of $H$ and $\ket{\lambda^-}$ be an associated eigenvector. Then, $\lambda^-$ is simple, $\braket{\lambda^-}{\phi_0}\neq 0$, and $\lambda^-< -\gamma \phi_0$ for any $\gamma>0$.    
\end{prop}
\begin{proof}
Since $-H$ is irreducible and nonnegative, by the Perron--Frobenius theorem (see, e.g.,~\cite{BCN1989B,BH2012B,HJ13}), $\lambda^-$ is simple and the entries of $\ket{\lambda^-}$ are all nonnegative or all nonpositive.
In particular, $\braket{\lambda^-}{\phi_0}\neq 0$.
Moreover, since $\gamma A$ and $-H-\gamma A=\sum_{w\in W}\ket{w}\bra{w}$ are both nonnegative and the latter is nonzero, we have $\lambda^-<-\gamma\phi_0$ for any $\gamma>0$.
\end{proof}

Let $\lambda^+$ denote the eigenvalue of $H$ that is the $(|W|+1)$-smallest, considering the multiplicity of each eigenvalue. If $\lambda^+$ is in the spectrum $\sigma(-\gamma A)$, or if the eigenvectors $\ket{\lambda^+}$ corresponding to $\lambda^+$ satisfy $\braket{\lambda^+}{\phi_0}=0$ for $\gamma>0$, meaning that $\ket{\lambda^+}$ is derived from a $0$-eigenvector of $M^\lambda$ fulfilling condition~\eqref{eq:cond_psi0}, then we must select the next smallest eigenvalue. We claim that there exists a specific value of $\gamma$ for which the eigenvalue $-\gamma\phi_0$ of $-\gamma A$ is positioned midway between $\lambda^-$ and $\lambda^+$. Our approach uses $\lambda^\pm$ and this $\gamma$ in the asymptotic limit. Essentially, for any graph $G(V,E)$ we claim the existence of a $\gamma$ value satisfying
\begin{equation*}
\lambda^{\pm} = -\gamma\phi_0 \pm \epsilon,
\end{equation*}
where $|\lambda^- + \gamma\phi_0|=|\lambda^+ + \gamma\phi_0|=\epsilon$.

Before we present the proof of this claim, let us first offer some intuitive insights into its validity. Considering $H$ as defined in Eq.~\eqref{eq:hamiltonian}, when $\gamma=0$, the eigenvalues of $H$ are $-1$ (with multiplicity $|W|$) and $0$ (with multiplicity $N-|W|$), where $N$ is the dimension of the Hilbert space. For small, nonzero $\gamma$, the eigenvalue of $-1$ splits into $|W|$ eigenvalues (which are not necessarily distinct) corresponding to different eigenvectors. Note that, based on a dimensionality analysis, at least one of these $|W|$ eigenvectors $\ket{\lambda}$ satisfies $\braket{\lambda}{\phi_0} \neq 0$, aligned with $\braket{\lambda^-}{\phi_0} \neq 0$. In the case when $|W|\ge 2$, at least two eigenvalues will be different because $\lambda^-$, the smallest eigenvalue, is simple. As previously described, $\lambda^+$ represents the $(|W|+1)$-th eigenvalue of $H$, considering the multiplicity of each eigenvalue. Consequently, $\lambda^+$ is selected as the smallest eigenvalue of $H$ that approaches 0 when $\gamma$ tends to 0. Due to the stipulation in constraint~\eqref{eq:cond_psi0} and the condition $\lambda^+\notin \sigma(-\gamma A)$, there are instances where $\lambda^+$ needs to be the $j$-th eigenvalue, where $j$ exceeds $|W|+1$. It is essential to avoid any eigenvalues that converge towards $-1$ as $\gamma$ decreases when choosing $\lambda^+$ because, in such scenarios, the gap between $\lambda^-$ and $\lambda^+$ becomes too small.

When $\gamma$ is small, the eigenvalue $-\gamma \phi_0$ of $-\gamma A$ is near zero, placing it close to $\lambda^+$ on the real line spectrum. Conversely, as $\gamma$ increases, $-\gamma\phi_0$ becomes increasingly aligned with $\lambda^-$, indicating an one-to-one correlation between the eigenvalues of $-\gamma A$ and $H$. We have established that for small $\gamma$, the eigenvalue $-\gamma \phi_0$ is positioned near the right of the range $[\lambda^-,\lambda^+]$. As $\gamma$ increases, this eigenvalue shifts leftward, eventually nearing the left end of this range. Given that the determinants $|H-\lambda I|$ and $|-\gamma A-\phi I|$ are continuous real functions of $\gamma$, by applying the intermediate value theorem, we deduce that a specific $\gamma$ exists where $-\gamma \phi_0$ is exactly at the midpoint of $[\lambda^-,\lambda^+]$.

\begin{prop}\label{prop:optimal-gamma}
Let $\lambda^-$ be the smallest eigenvalues of $H$. Let $\lambda^+$ be an eigenvalue of $H$ such that $\lambda^+> \lambda^-$. Then, there exists $\gamma>0$ such that 
 \[-\gamma\phi_0 = \frac{\lambda^-+\lambda^+}{2},\]
where $\phi_0$ is the largest eigenvalue of $A$.
\end{prop}
\begin{proof} 
For convenience, set $\tau=\gamma/(1+\gamma)$.
Note that $\tau\in(0,1)$.
Also, let $\mu^{\pm}(\tau)=\lambda^{\pm}(\gamma)/(1+\gamma)$, so that they are eigenvalues of
\begin{equation*}
   (1-\tau)H(\gamma)= -\tau A-(1-\tau)\sum_{w\in W}\ket{w}\bra{w}.
\end{equation*}
Then, observe that the values $\mu^{\pm}(\tau)$ also make sense for $\tau=0,1$.
In particular, 
\begin{equation*}
    \mu^-(0)=-1, \qquad \mu^+(0)\leqslant 0, \qquad \mu^-(1)=-\phi_0, \qquad \mu^+(1)=-\phi_{\ell}
\end{equation*}
for some $\ell\geqslant 1$, where $\phi_{\ell}$ is the $(\ell+1)^{\mathrm{st}}$ largest eigenvalue of $A$.
We have used that $\phi_0$ is a simple eigenvalue of $A$ because $G(V,E)$ is connected.
Moreover, note that $\mu^{\pm}(\tau)$ are continuous in $\tau$ using Weyl's Perturbation Theorem~\cite{Bha97}.
Now, consider the continuous function \begin{equation*}
    h(\tau)=\mu^-(\tau)+\mu^+(\tau)+2\phi_0\tau.
\end{equation*}
We have $h(0)\leqslant -1<0$ and $h(1)=\phi_0-\phi_{\ell}>0$.
Hence, by the intermediate value theorem, there exists $\tau\in(0,1)$ such that $h(\tau)=0$.
The proof is complete after converting back to variable $\gamma=\tau/(1-\tau)$.
\end{proof}

To illustrate an eigenvalue $\lambda$ to avoid because it is in $\sigma(-\gamma A)$, consider the complete bipartite graph $K_{n,n}$ with $n$ marked vertices in one part. The eigenvalues of $H$ are ordered as $\lambda^-<-1<0<\lambda^+$. The second smallest eigenvalue is $-1$ with a multiplicity of $n-1$. The next eigenvalue, $0$, belongs to $\sigma(-\gamma A)$, making $\lambda^+$ the $(n+2)$-th and final eigenvalue of $H$.



\subsection{Asymptotic analysis}\label{asymptotic}

Let us consider a graph class characterized by some parameters that determine the number of vertices $N$. Usually, the spectral gap of the Hamiltonian tends to zero when the number of vertices $N$ increases. In this case, the analysis of the search algorithm is simpler in the asymptotic regime.

In cases where one vertex is marked, several studies have employed a technique to analyze quantum-walk-based search algorithms. This technique uses two key eigenvectors of the evolution operator, ensuring significant overlap between both the initial and marked states and the space spanned by these eigenvectors~\cite{CG_04,AKR05,Gro97}. In this work, we extend and generalize this technique to cases with multiple marked vertices. The asymptotic computational complexity of the search algorithm primarily depends on the eigenvalues $\lambda^{\pm}$, where $\lambda^-$ is the smallest eigenvalue of $H$, and $\lambda^+$ is the $(|W|+1)$-smallest, taking into account the multiplicity of each eigenvalue. Besides, $\lambda^+$ must satisfy the conditions $\lambda^+\notin\sigma(-\gamma A)$ and $\braket{\lambda^+}{\phi_0}\neq 0$ for $\gamma>0$. The value of $\gamma$ is described in Proposition~\ref{prop:optimal-gamma}, which positions $-\gamma\phi_0$ in the middle of the range $[\lambda^-,\lambda^+]$. Since calculating this optimal $\gamma$ might be too difficult, we describe an alternative way to obtain $\gamma$, which asymptotically approaches the optimal one. Thus, we can write
\begin{equation}\label{eq:lambdapm}
\lambda^{\pm} = -\gamma\phi_0 \pm \epsilon + o(\epsilon). 
\end{equation}
In this case, the gap between $\lambda^-$ and $\lambda^+$ is asymptotically $2\epsilon>0$. Note that $2\epsilon$ is larger than the spectral gap of $H$ unless $|W|=1$. We will detail a method to calculate $\gamma$ needed in Eq.~\eqref{eq:lambdapm} below.

Using Eq.~\eqref{eq:Mww}, matrices $M^{\lambda^\pm}$ are given by
\begin{equation*}
    M_{ww'}^{\lambda^{\pm}} = \delta_{ww'} +\sum_{\ell=0}^k \frac{\bra{w}P_{\ell}\ket{w'}}{\gamma(\phi_{\ell}-\phi_0)\pm {\epsilon} + o(\epsilon)}.
\end{equation*}
Up to second order in $\epsilon$ we obtain
\begin{equation}\label{eq:Mapprox}
    M_{ww'}^{\lambda^{\pm}} = \pm\frac{\bra{w}P_0\ket{w'}}{\epsilon} + \delta_{ww'}  - \frac{S^{(1)}_{ww'}}{\gamma} \mp \frac{\epsilon S^{(2)}_{ww'}}{\gamma^2}+O(\epsilon^2),
\end{equation}
where
\begin{equation*}
    S^{(1)}_{ww'} = \sum_{\ell=1}^k \frac{\bra{w}P_{\ell}\ket{w'}}{\phi_0-\phi_{\ell}}
\end{equation*}
and 
\begin{equation*}
    S^{(2)}_{ww'} = \sum_{\ell=1}^k \frac{\bra{w}P_{\ell}\ket{w'}}{(\phi_0-\phi_{\ell})^2}.
\end{equation*}
Note that the sums $S^{(1)}_{ww'}$ and $S^{(1)}_{ww'}$, as well as $\bra{w}P_0\ket{w'}$, can be calculated if we know the eigenvalues and eigenprojections of $A$. The unknown quantities in Eq.~\eqref{eq:Mapprox} are $\gamma$ and $\epsilon$.

One of the goals of the asymptotic analysis is to determine analytical values of $\gamma$ and $\epsilon$. By using Eq.~\eqref{eq:Mapprox}, the condition $\det(M^{\lambda^\pm})=0$, and considering the numerator up to the order of $\epsilon^2$, we derive a quadratic equation in the form $\epsilon^2+a(\gamma)\epsilon-b(\gamma)=O(\epsilon^3)$. This equation cannot contain the linear term, as this is the only way to achieve the $\pm \epsilon$ term in Eq.~\eqref{eq:lambdapm}. Consequently, we use the equation $a(\gamma)=0$ to ascertain the value of $\gamma$, which is asymptotically optimal. Once $\gamma$ is computed, $\epsilon$ can be determined as $\sqrt{b(\gamma)}$.

Proposition~\ref{prop:optimal-gamma} ensures that Eq.~\eqref{eq:lambdapm} is applicable to any graph $G(V,E)$. Furthermore, this asymptotic analysis is valid for any graph such that $\epsilon\ll 1$ or for any graph class where the gap between $\lambda^-$ and $\lambda^+$ approaches 0 as $N$ tends toward infinity. An example where the gap between $\lambda^-$ and $\lambda^+$ does not approach zero is the class of complete bipartite graphs $K_{n,n}$ with $n$ marked vertices in one part. In this case, using that $\phi_0=n$, the optimal parameters are $\lambda^\pm=(-1\pm\sqrt{2})/2$ and $\gamma=1/2n$. This yields a fixed $\epsilon=1/\sqrt{2}$. Note that the behavior of the gap depends on the placement of the marked vertices. In the same graph $K_{n,n}$, if the marked vertices are arranged in a different configuration than the one above, the gap approaches zero.

\subsection{Calculation of the computational complexity}\label{subssec:compcomplexity}

Using an orthonormal basis ${\ket{\lambda}}$ of eigenvectors of $H$, we can express the probability of finding a marked vertex as a function of time $t$ (refer to Eq.~\eqref{eq:p(t)}) as follows:
\begin{equation*}
p_\text{exact}(t) = \sum_{w\in W}\Big| \sum_\lambda \mathrm{e}^{-\mathrm{i}\lambda t} \braket{w}{\lambda}\braket{\lambda}{\psi(0)}\Big|^2.
\end{equation*}
The exact eigenvalues $\lambda$ correspond to the roots of $\det(M^\lambda)=0$. However, calculating these roots and the exact expression for $p_\text{exact}(t)$ often proves too complex. In the asymptotic regime, where $N \rightarrow \infty$ and $\epsilon \rightarrow 0$, it is preferable to use Eq.~\eqref{eq:Mapprox} as long as the computational complexity is determined by $\lambda^\pm$. Matrices $M^{\lambda^{\pm}}$ are simpler to compute as their entries depend only on two sums and the term $\bra{w}P_{0}\ket{w'}$. To calculate an approximate $p_\text{approx}(t)$ in this regime, we focus on the dominant terms:
\begin{equation}\label{eq:p(t)-2terms}
p_\text{approx}(t) = \sum_{w\in W}\left| \mathrm{e}^{\mathrm{i}\epsilon t} \braket{w}{\lambda^-}\braket{\lambda^-}{\psi(0)}+ \mathrm{e}^{-\mathrm{i}\epsilon t} \braket{w}{\lambda^+}\braket{\lambda^+}{\psi(0)}\right|^2.
\end{equation}
$\ket{\lambda^-}$ is a normalized eigenvector of $H$ associated with the smallest root $\lambda^-$ of $\det(M^{\lambda}) = 0$ and
$\ket{\lambda^+}$ with the $(|W|+1)$-smallest root $\lambda^+$, subjected to the restrictions $\lambda^+\notin\sigma(-\gamma A)$ and $\braket{\lambda^+}{\phi_0}\neq 0$. 
Eigenvalue $\lambda^-$ is always simple but if eigenvalue $\lambda^+$ is not simple, which is uncommon, we must replace $\ket{\lambda^+}\bra{\lambda^+}$ with the eigenprojection $P^+ = \sum \ket{\lambda^+}\bra{\lambda^+}$ in equation~\eqref{eq:p(t)-2terms}, where the sum runs over an orthonormal basis of the $\lambda^+$-eigenspace. The multiplicity of $\lambda^+$ is at most $|W|$.

Note that $p_\text{exact}(t)\ge p_\text{approx}(t)$. The terms of $p_\text{approx}(t)$ are calculated as follows. The value of $\gamma$ is derived from the equation $a(\gamma) = 0$ and $\epsilon$ from $\sqrt{b(\gamma)}$, as described in Sec.~\ref{asymptotic}. Now, let us turn our attention to calculating $\big|\braket{w}{\lambda^+}\big|$. The approach for calculating $\big|\braket{w}{\lambda^-}\big|$ follows a similar pattern. Remember that $\braket{w}{\lambda^+}\bigr\rvert_{w\in W}$ is a $0$-eigenvector of $M^{\lambda^+}$, with entries $\braket{w}{\lambda^+}$ for each $w\in W$. Given that $M^{\lambda^+}$ is already determined, we can identify a normalized $0$-eigenvector of $M^{\lambda^+}$, which has entries denoted as $u(w)$ for each $w\in W$. In cases where $\lambda^+$ is not a simple eigenvalue, it becomes necessary to establish an orthonormal basis for the kernel of $M^{\lambda^+}$. Following this,
\begin{equation}\label{wlam}
    \braket{w}{\lambda^+}={c^+}\, u(w),
\end{equation}
where ${c^+}$ is a complex number. Using $\sum_{\ell} \left\| P_{\ell} \ket{\lambda^+} \right\|^2=1$
and Eq.~\eqref{masterequation}, we obtain 
\begin{equation*}
    \frac{1}{\left|{c^+}\right|^2} = \sum_{\ell=0}^k \frac{1}{\left|\lambda^++\gamma\phi_{\ell}\right|^2}\,\Big\|\sum_{w\in W} u(w) P_{\ell}\ket{w}\Big\|^2.
\end{equation*}
The dominant term is
\begin{equation*}
    \frac{1}{\left|{c^+}\right|} =  \frac{1}{\epsilon}\,\Big\|\sum_{w\in W} u(w) P_{0}\ket{w}\Big\|+O(1).
\end{equation*}
To calculate $\braket{w}{\lambda^+}$ instead of its absolute value, we still have to find ${c^+}$ instead of $\left|{c^+}\right|$. This is easily solved by rescaling $\ket{\lambda^+}$ with a unit complex number $\ket{\lambda^+}\rightarrow \mathrm{e}^{\mathrm{i}\theta}\ket{\lambda^+}$ so that ${c^+} \mathrm{e}^{-\mathrm{i}\theta}$ is a positive real number. After rescaling $\ket{\lambda^+}$, we have $\braket{w}{\lambda^+}=|c^+|\,u(w)$.

The last terms of $p_\text{approx}(t)$ we need to determine are $\big|\braket{\lambda^\pm}{\psi(0)}\big|$. We define $\ket{\psi(0)}$ as a normalized eigenvector of $-\gamma A$ corresponding to the eigenvalue $-\gamma \phi_0$. In the special case where the graph $G(V,E)$ is regular, $\ket{\psi(0)}$ is the uniform superposition and thus becomes an eigenvector with the eigenvalue $-\gamma d$, where $d$ is the degree of the graph. Setting $\ell=0$ and applying Eq~\eqref{masterequation}, we derive
\begin{equation}\label{eq:main2}
\braket{\psi(0)}{\lambda^\pm} =\mp  \frac{1}{\epsilon}\sum_{w\in W}\braket{\psi(0)}{w}\braket{w}{\lambda^\pm}.
\end{equation}
Given that all terms on the right-hand side are assumed to be known, we now have a method to calculate $\braket{\lambda^\pm}{\psi(0)}$.

Note that we have calculated all terms of $p_\text{approx}(t)$ as given by Eq.~\eqref{eq:p(t)-2terms}. These terms are obtained from Eqs.~\eqref{wlam} and~\eqref{eq:main2}, while $\epsilon$ and $\gamma$ are calculated using the method described in Sec.~\ref{asymptotic}.  The framework described here avoids the need to calculate full expressions for the eigenvectors $\ket{\lambda^\pm}$. To determine the computational complexity of the search algorithm, we still need to identify the optimal running time, which is the value of $t_\text{opt}$ that corresponds to the first peak of $p_\text{approx}(t)$, and the corresponding success probability, which is $p_\text{approx}(t_\text{opt})$.

The effectiveness of this framework hinges on the substantial overlap of both the initial state and the marked states with the space spanned by $\ket{\lambda^\pm}$. To analyze this, we can set a small upper bound for the terms not included in Eq.~\eqref{eq:p(t)-2terms} concerning the initial condition. These missing terms are upper-bounded by
\begin{align*}
\sum_{\lambda\neq \lambda^\pm}\big|\braket{\psi(0)}{\lambda} \big|^2 \leqslant \frac{1}{\left(\lambda^\circ+\gamma\phi_0\right)^2}\sum_{\lambda\neq \lambda^\pm}\left(\lambda+\gamma\phi_0\right)^2\big|\braket{\psi(0)}{\lambda}\big|^2,
\end{align*}
where $\lambda^\circ$ represents the eigenvalue of $H$ immediately following $\lambda^-$. Using Eq.~\eqref{masterequation} and taking $\ell=0$, we find $\left(\lambda+\gamma\phi_0\right)\braket{\psi(0)}{\lambda}=-\sum_w \braket{\psi(0)}{w}\braket{w}{\lambda}$. This leads to the upper bound
\begin{align*}
\sum_{\lambda\neq \lambda^\pm}\big|\braket{\psi(0)}{\lambda} \big|^2 \leqslant \frac{|W|+2}{\left(\lambda^\circ+\gamma\phi_0\right)^2} \sum_{w,w'\in W}\bra{w}P_0\ket{w'}.
\end{align*}
Typically, the right-hand side tends to decrease. For example, in regular graphs, the summation on the right-hand side simplifies to $|W|^2/N$. If the gap between the eigenvalue $\lambda^\circ$ and $-\gamma\phi_0$ is sufficiently large, such that it diminishes more slowly than $1/\sqrt{N}$, then the overlap $|\braket{\psi(0)}{\lambda}|^2$ for the excluded eigenvectors in Eq.~\eqref{eq:p(t)-2terms} approaches zero. If the error approaches zero asymptotically, we can state that $p_\text{exact}(t)=p_\text{approx}(t)+o(1)$ and the framework described in this work is successful.

There is an interesting case where the success probability $p_\text{approx}(t)$ is the square of a sinusoidal function. This occurs when the two terms within the summation of Eq.~\eqref{eq:p(t)-2terms} fulfill the following condition
\begin{equation}\label{bra_lam_w}
\braket{\lambda^+}{\psi(0)}\braket{w}{\lambda^+} = -\braket{\lambda^-}{\psi(0)}\braket{w}{\lambda^-} + o(1)
\end{equation}
for every $w\in W$. Indeed, by applying Eqs.~\eqref{eq:p(t)-2terms} and~\eqref{bra_lam_w}, we obtain
\begin{equation}\label{eq:p(t)=sin^2}
    p_\text{approx}(t)=4\left|\braket{\lambda^-}{\psi(0)}\right|^2\sum_{w\in W} \left|\braket{w}{\lambda^-}\right|^2\sin^2{\epsilon t} + o(1) + o(\epsilon t).
\end{equation}
From this, we can determine the optimal running time as
\begin{equation}\label{final_t_run}
    t_{\mathrm{run}} = \frac{\pi}{2\epsilon},
\end{equation}
which shows that the gap between $\lambda^\pm$ determines the optimal running time. A larger gap means a smaller runtime. 
This expression leads to the success probability being
\begin{equation}\label{final_p_succ}
    p_{\mathrm{succ}} = 4\left|\braket{\lambda^-}{\psi(0)}\right|^2\sum_{w\in W} \left|\braket{w}{\lambda^-}\right|^2 + o(1).
\end{equation} 
In the next section, we confirm that condition~\eqref{bra_lam_w} applies to a quantum walk on the Johnson graph with two marked vertices. This condition is also observed in cases with one marked vertex~\cite{CG_04,TSP22} and multiple marked vertices~\cite{Won16b,LP23}. The results of those papers in terms of computational complexity can be fully reproduced using the asymptotic results of this subsection, in particular expressions~\eqref{final_t_run} and~\eqref{final_p_succ}.

\section{Johnson graphs}\label{sec:JohnsonGraph}

For the remainder of this paper, we will apply the general formula outlined in the previous section to the Johnson graph $G(V, E) = J(n, k)$, with two distinct marked vertices. The computation can be expanded to accommodate more than two marked vertices; however, the complexity of the calculations rapidly escalates.
The vertex set $V$ is the set of $k$-subsets of $[n]=\{1,2,\dots,n\}$, and two vertices $v,v'\in V$ are adjacent if and only if $|v\cap v'|=k-1$.
Note that $J(n,1)$ is the complete graph $K_n$, and that $J(n,2)$ is the triangular graph $T_n$, which is strongly regular.
We will assume that $W=\{w_1,w_2\}$, where we fix $k$ and let $n\rightarrow\infty$.
In this case, the matrix $M^{\lambda}$ from Eq.~\eqref{eq:Mww} is a two-dimensional matrix denoted by
\begin{equation*}
	M^{\lambda}= \begin{bmatrix}
		m_1 & m_3 \\
		m_3 & m_2 \\
	\end{bmatrix},
\end{equation*}
where 
\begin{align*}
m_1 &= \sum_{\ell=0}^k \frac{ \left\| P_{\ell}\ket{w_1}\right\|^2 }{\lambda+\gamma \phi_{\ell}}+1,\\
m_2 &= \sum_{\ell=0}^k \frac{ \left\| P_{\ell}\ket{w_2}\right\|^2 }{\lambda+\gamma \phi_{\ell}}+1,\\
m_3 &= \sum_{\ell=0}^k \frac{ \bra{w_1}P_{\ell}\ket{w_2} }{\lambda+\gamma \phi_{\ell}}.
\end{align*}
The eigenvalues $\lambda$ of $H$ are obtained from $\det(M^{\lambda})=m_1m_2-m_3^2=0$.

The number of vertices of $J(n,k)$ is $N=\binom{n}{k}$.
The $k+1$ distinct eigenvalues $\phi_0>\phi_1>\dots>\phi_k$ of $J(n,k)$ are given by
\begin{equation}\label{eigenvalues}
	\phi_{\ell}=(k-\ell)(n-k-\ell)-\ell,
\end{equation}
and the multiplicity of $\phi_{\ell}$ is $\binom{n}{\ell}-\binom{n}{\ell-1}$ (with the convention that $\binom{n}{-1}=0$).
See~\cite{BI1984B,BCN1989B,DKT2016EJC}.
For the Johnson graph $J(n,k)$, it important to remark that $P_{\ell}$ has constant diagonal entries $\big(\binom{n}{\ell}-\binom{n}{\ell-1}\big)/\binom{n}{k}$, so that
\begin{equation}\label{Pw}
	\left\| P_{\ell} \ket{w_1}\right\|^2=\left\| P_{\ell} \ket{w_2}\right\|^2 =\frac{\binom{n}{\ell}-\binom{n}{\ell-1}}{\binom{n}{k}}=\frac{k!(n-k)!(n-2\ell+1)}{\ell!(n-\ell+1)!}.
\end{equation}
In particular, we have $m_1=m_2$.
Let
\begin{equation*}
	\delta=k-|w_1\cap w_2|,
\end{equation*}
which is the distance between $w_1$ and $w_2$.
We note that $m_3$ depends on $\delta$. 
More specifically, it is known that $\bra{w_1}P_{\ell}\ket{w_2}$ is written in terms of a terminating $_3F_2$ hypergeometric series:
\begin{equation}\label{hypergeometric}
	\bra{w_1}P_{\ell}\ket{w_2} =  \left\| P_{\ell}\ket{w_1}\right\|^2 \hypergeometricseries{3}{2}{-\ell,-\delta,\ell-n-1}{k-n,-k}{1}.
\end{equation}
See, e.g.,~\cite[pp.~219--220]{BI1984B} and~\cite[Example 2.3]{MT2009EJC}.
From now on, \emph{we assume that $\delta$ is known in advance.}
After analyzing the search algorithm under this assumption, we will discuss the unrestricted case, in which $\delta$ is unknown.

\subsection{An invariant subspace}

For the single-marked case discussed in~\cite{TSP22}, a $(k+1)$-dimensional invariant subspace of $\mathcal{H}$ played an important role.
For the present case, we can again make use of a similar but more complicated invariant subspace.
For integers $a,b,c$ such that $0\leqslant a\leqslant k-\delta$, $a\leqslant b\leqslant c$, and $2k-n+\delta+a-b\leqslant c\leqslant \min\{k+a-b,\delta+a\}$, let
\begin{equation*}
    \nu_{a,b,c} = \big\{ v\in V : |v\cap w_1\cap w_2|=a,\,\{|v\cap w_1|,|v\cap w_2|\}=\{b,c\} \big\},
\end{equation*}
and set
\begin{equation*}
    \ket{\nu_{a,b,c}} = \sum_{v\in \nu_{a,b,c}} \ket{v}.
\end{equation*}
Let $\mathcal{H}_{\mathrm{inv}}$ denote the subspace of $\mathcal{H}$ spanned by these (unnormalized) vectors $\ket{\nu_{a,b,c}}$.
Then, it follows that $\mathcal{H}_{\mathrm{inv}}$ is invariant under $A$.
This can be shown manually, but the group-theoretic explanation is as follows.
When $n>2k$, the automorphism group of $J(n,k)$ is isomorphic to the symmetric group $\mathfrak{S}_n$ on $n$ letters $[n]=\{1,2,\dots,n\}$.
Consider the subgroup $\mathfrak{G}$ of $\mathfrak{S}_n$ consisting of the elements that fix $W$ setwise.
The elements of $\mathfrak{G}$ fix each of the sets $w_1\cap w_2$ and $[n]\setminus (w_1\cup w_2)$, and either fix or swap the two sets $w_1\setminus w_2$ and $w_2\setminus w_1$.
Hence, the order of $\mathfrak{G}$ is $|\mathfrak{G}|=(k-\delta)!\times (n-k-\delta)!\times (\delta!)^2\times 2$.
The $\nu_{a,b,c}$ are precisely the orbits of $\mathfrak{G}$ on the vertex set $V$.
In particular, this implies that a vector $\mathbf{v}\in\mathcal{H}$ is $\mathfrak{G}$-invariant (i.e., $g\mathbf{v}=\mathbf{v}$ for all $g\in\mathfrak{G}$) if and only if $\mathbf{v}$ is a linear combination of the vectors $\ket{\nu_{a,b,c}}$, or equivalently, $\mathbf{v}\in\mathcal{H}_{\mathrm{inv}}$.
For every $\mathbf{v}\in\mathcal{H}_{\mathrm{inv}}$, since $g\mathbf{v}=\mathbf{v}$ and $Ag=gA$ for all $g\in\mathfrak{G}$, we have
\begin{equation*}
    g(A\mathbf{v})=A(g\mathbf{v})=A\mathbf{v} \qquad (g\in\mathfrak{G}),
\end{equation*}
so that $A\mathbf{v}$ is $\mathfrak{G}$-invariant and thus belongs to $\mathcal{H}_{\mathrm{inv}}$.
We have now shown that $\mathcal{H}_{\mathrm{inv}}$ is invariant under $A$.
Moreover, since $W=\nu_{k-\delta,k-\delta,k}$, we have
\begin{equation*}
	\left(\sum_{w\in W} \ket{w} \bra{w}\right) \! \ket{\nu_{a,b,c}} = \delta_{k-\delta,a}\delta_{k-\delta,b}\delta_{k,c} \ket{\nu_{k-\delta,k-\delta,k}}.
\end{equation*}
Hence, it follows that $\mathcal{H}_{\mathrm{inv}}$ is invariant under $H$ and thus also $\mathrm{e}^{-\mathrm{i}Ht}$.
Note also that, 
since the sets $\nu_{a,b,c}$ partition $V$,
the sum of the vectors $\ket{\nu_{a,b,c}}$ equals $\sqrt{N}\ket{\psi(0)}$, so that $\ket{\psi(0)}$ belongs to $\mathcal{H}_{\mathrm{inv}}$.

From now on, \emph{we will always consider eigenvectors $\ket{\lambda}$ in $\mathcal{H}_{\mathrm{inv}}$.}
This assumption provides us with the following strong constraint:
\begin{equation*}
    \braket{w_1}{\lambda}=\braket{w_2}{\lambda} = \frac{\braket{\nu_{k-\delta,k-\delta,k}}{\lambda}}{2}.
\end{equation*}
For the rest of this section, let us call this common value $\alpha$.
The following is a strengthening of  Proposition~\ref{prop:lambda_in_sigma_A} in this case. 

\begin{prop}\label{important equivalence}
The following are equivalent:
(i) $\lambda\in\sigma(-\gamma A)$; (ii) $-\gamma A\ket{\lambda}=\lambda\ket{\lambda}$; (iii) $\alpha=0$.
\end{prop}

\begin{proof}
(iii)$\Rightarrow$(ii): Note that $-\gamma A\ket{\lambda}=H\ket{\lambda}=\lambda\ket{\lambda}$ by the definition of $H$.

(ii)$\Rightarrow$(i): Clear.

(i)$\Rightarrow$(iii): Suppose that $\lambda=-\gamma\phi_{\ell}$.
Then, by Eq.~\eqref{before master equation}, we have
\begin{equation*}
    \alpha \left(P_{\ell}\ket{w_1}+P_{\ell}\ket{w_2}\right)=0,
\end{equation*}
where we note that $P_{\ell}A=AP_{\ell}=\phi_{\ell}P_{\ell}$.
If $\ell=0$ then $P_0\ket{w_1}=P_0\ket{w_2}=(1/\sqrt{N})\ket{\psi(0)}$, so $\alpha=0$.
If $\ell>0$ then it follows from the fact that $J(n,k)$ is a primitive\footnote{A distance-regular graph with diameter $k$ is called \emph{primitive} if all the distance-$i$ graphs $(i=1,2,\dots,k)$ are connected.} distance-regular graph (provided $n>2k$) that $P_{\ell}\ket{w_1}$ and $P_{\ell}\ket{w_2}$ are linearly independent (cf.~\cite[Sec.~9.1]{BCN1989B},~\cite[p.~137]{BI1984B}), so again we have $\alpha=0$, whence (iii).
\end{proof}

\begin{corollary}\label{- gamma d is not eigenvalue}
$-\gamma\phi_0$ is not an eigenvalue of $H$ on $\mathcal{H}_{\mathrm{inv}}$.
\end{corollary}

\begin{proof}
If $\lambda=-\gamma\phi_0$ then $\ket{\psi(0)}$ and $\ket{\lambda}$ are linearly independent eigenvectors of $A$ with eigenvalue $\phi_0$.
(The linear independence follows from $\braket{w_1}{\psi(0)}=\braket{w_2}{\psi(0)}=1/\sqrt{N}\ne 0$ and $\alpha=0$.)
But $\phi_0$ has multiplicity one since $J(n,k)$ is connected, a contradiction.
\end{proof}

\begin{corollary}\label{zero of m1+m3}
Suppose that $\lambda\not\in\sigma(-\gamma A)$.
Then, $m_1+m_3=m_2+m_3=0$.
\end{corollary}

\begin{proof}
Immediate from Eq.~\eqref{0-eigenvector} and $\braket{w_1}{\lambda}=\braket{w_2}{\lambda}=\alpha\ne 0$. 
\end{proof}

\bigskip

This corollary simplifies the discussions regarding the calculations of $\lambda^{\pm}$, as it allows us to disregard the scenario where $m_1-m_3=0$, which leads to a {``bad'' eigenvalue} (we will revisit this point later). Additionally, note that the findings of this subsection are extendable to other families of distance-transitive graphs, such as Hamming graphs.

\subsection{Calculation of $\lambda^{\pm}$ and related quantities}

First, we find the eigenvalue $\lambda^-$ of the ground state of $H$.
Note that the adjacency matrix $A$ and the oracle $\sum_{w\in W}\ket{w}\bra{w}$ are both nonnegative matrices with respect to the computational basis, and hence so are their matrix representations with respect to the $\ket{\nu_{a,b,c}}$ on $\mathcal{H}_{\mathrm{inv}}$.
Since $\gamma>0$, it follows from the Perron--Frobenius theorem (see, e.g.,~\cite{BCN1989B,BH2012B,HJ13}) that the eigenvalue $\lambda^-$ is in the interval $(-\infty,-\gamma\phi_0)$ and has multiplicity one; {see also Proposition~\ref{from PF}}.
Moreover, observe that $m_1+m_3$, as a function of $\lambda$, is monotone decreasing on $(-\infty,-\gamma\phi_0)$ since $\left|\bra{w_1}P_{\ell}\ket{w_2}\right|\leqslant \left\| P_{\ell}\ket{w_1}\right\|^2=\left\| P_{\ell}\ket{w_2}\right\|^2$ by the Cauchy--Schwarz inequality and Eq.~\eqref{Pw}.
Hence, it follows from Corollary~\ref{zero of m1+m3} that $\lambda^-$ is the \emph{unique} eigenvalue of $H$ on $\mathcal{H}_{\mathrm{inv}}$ in $(-\infty,-\gamma\phi_0)$.

We write
\begin{gather*}
	S_1=S_{w_1w_1}^{(1)}=S_{w_2w_2}^{(1)}, \qquad S_1'=S_{w_1w_2}^{(1)}=S_{w_2w_1}^{(1)}, \\
	S_2=S_{w_1w_1}^{(2)}=S_{w_2w_2}^{(2)}, \qquad S_2'=S_{w_1w_2}^{(2)}=S_{w_2w_1}^{(2)},
\end{gather*}
and
\begin{equation*}
	\lambda=-\gamma\phi_0+\epsilon.
\end{equation*}
Using the equations of subsection~\ref{asymptotic}, $m_1 + m_3$ is given by 
\begin{equation}\label{m1+m3}
    m_1 + m_3 = \frac{2}{N\epsilon} + 1 - \frac{S_1 + S_1'}{\gamma} - \frac{S_2 + S_2'}{\gamma^2}\,\epsilon + O\!\left(\epsilon^2\right).
\end{equation}
We choose $\gamma$ so that $\epsilon(m_1+m_3)$ has no linear term in $\epsilon$, i.e.,
\begin{equation}\label{gamma}
	\gamma = S_1+S'_1.
\end{equation}
We note that $S'_1$ depends on the distance $\delta$, which we currently assume is known in advance.

To proceed with the calculations, we need the following expressions for the sums $S_1,S_2,S'_1$, and $S'_2$:
\begin{gather}
	S_1 = \frac{1}{kn}+O\!\left(\frac{1}{n^2}\right), \qquad S_2 = \frac{1}{k^2n^2}+O\!\left(\frac{1}{n^3}\right), \label{S} \\
	S'_1 = O\!\left(\frac{1}{n^{\delta+1}}\right), \qquad S'_2 = O\!\left(\frac{1}{n^{\delta+2}}\right). \label{S'}
\end{gather}
Eq.~\eqref{S} follows easily from Eqs.~\eqref{eigenvalues} and \eqref{Pw}.
Indeed, for $0\leqslant\ell\leqslant k$, we have
\begin{equation}\label{eigenvalue difference}
    \phi_0-\phi_{\ell}=\ell(n-\ell+1)=\ell n+O(1),
\end{equation}
and
\begin{equation*}
    \|P_{\ell}\ket{w_1}\|^2=\frac{k!(n-2\ell+1)}{\ell!(n-\ell+1)\cdots(n-k+1)}=\frac{k!}{\ell!n^{k-\ell}}+O\!\left(\frac{1}{n^{k-\ell+1}}\right).
\end{equation*}
The proof of Eq.~\eqref{S'} is deferred to Appendix~\ref{sec: proofs of eqs}.
From Eqs.~\eqref{S} and \eqref{S'} it follows that
\begin{equation}\label{evaluate gamma}
	\gamma = \frac{1}{k n} + O\!\left(\frac{1}{n^2}\right).
\end{equation}

For convenience, set
\begin{equation*}
	\eta=\frac{1}{n}.
\end{equation*}
By Eq.~\eqref{Pw} {again}, we have
\begin{equation*}
    \left\| P_{\ell} \ket{w_1}\right\|^2 = \frac{k!\eta^{k-\ell}(1-(2\ell-1)\eta)}{\ell!(1-(\ell-1)\eta)\cdots(1-(k-1)\eta)} \qquad (0\leqslant\ell\leqslant k),  
\end{equation*}
which is analytic around $\eta=0$.
Moreover, it follows from Eq.~\eqref{hypergeometric} that $\bra{w_1}P_{\ell}\ket{w_2}$ is also analytic around $\eta=0$.
By Eqs.~{\eqref{eigenvalue difference}} and \eqref{evaluate gamma}, we have
\begin{equation}\label{denominator}
	\gamma(\phi_0-\phi_{\ell}) = \frac{\ell}{k} +O(\eta) \qquad (1\leqslant\ell\leqslant k),
\end{equation}
and this is also analytic around $\eta=0$.
Since the constant term of the right-hand side is nonzero, it follows from the above comments that $\epsilon(m_1+m_3)$, as a bivariate function of $\eta$ and $\epsilon=\lambda+\gamma\phi_0$, is analytic around the origin $(\eta,\epsilon)=(0,0)$.
Hence, the power series expansion of $\epsilon(m_1+m_3)$ centered at the origin converges absolutely and uniformly on some neighborhood of the origin (see~\cite[Sec.~2.3]{Krantz1992B}), and expressed in the form
\begin{equation*}
	\epsilon(m_1+m_3) = \frac{2}{N} - \frac{ S_2+S'_2 }{{\gamma}^{2}}\,{\epsilon}^{2} + O(\epsilon^3)
\end{equation*}
with respect to $\epsilon$, where we recall that the linear term in $\epsilon$ vanishes, and that
\begin{equation}\label{two comments}
	\frac{2}{N}=2k!\eta^k+O(\eta^{k+1}), \qquad \frac{ S_2+S'_2 }{{\gamma}^{2}} = 1+o(1)
\end{equation}
by Eqs.~{\eqref{S}, \eqref{S'}, and \eqref{evaluate gamma}}.
We note that the expression $O(\epsilon^3)$ in the above expansion is evaluated uniformly in $\eta$ on this neighborhood.
Now, set
\begin{equation*}
	\epsilon^-_0 = -\frac{ \sqrt{2}\, \gamma}{\sqrt{N(S_2+S'_2)}} = -\frac{\sqrt{2}}{\sqrt{N}}(1+o(1)).
\end{equation*}
Pick any scalar $\nu\in(1,2)$, and let $c\in(0,1)$ be such that $O(\epsilon^3)$ above is bounded as
\begin{equation*}
	|O(\epsilon^3)| < |\epsilon|^{1+\nu} 2^{-\nu}
\end{equation*}
for every $\epsilon\in(-c,c)$.
This is possible since $O(\epsilon^3)=o(\epsilon^{1+\nu})$ when $1+\nu<3$. Remember that this evaluation is uniform in $\eta$.
Let $n$ be large enough so that $|\epsilon^-_0|<c/2$ and $(S_2+S'_2)/\gamma^2>2/3$; {see Eq.~\eqref{two comments}}.
Then $\epsilon^-_0\pm |\epsilon^-_0|^{\nu}\in (-c,c)$ and we have
\begin{align*}
	\big|O\big((\epsilon^-_0\pm |\epsilon^-_0|^{\nu})^3\big)\big| &< \big| \epsilon^-_0\pm |\epsilon^-_0|^{\nu} \big|^{1+\nu}2^{-\nu} \\
	&< |\epsilon^-_0|^{\nu} \big| \epsilon^-_0\pm |\epsilon^-_0|^{\nu} \big| \\
	&< |\epsilon^-_0|^{\nu} \big| 2\epsilon^-_0\pm |\epsilon^-_0|^{\nu} \big|\cdot\frac{2}{3} \\
	&< \frac{S_2+S'_2}{\gamma^2} |\epsilon^-_0|^{\nu} \big| 2\epsilon^-_0\pm |\epsilon^-_0|^{\nu} \big|.
\end{align*}
Note that the right-hand side equals the absolute value of $2/N-(S_2+S'_2)\epsilon^2/\gamma^2$ for $\epsilon=\epsilon^-_0\pm |\epsilon^-_0|^{\nu}$.
We also note that $2/N-(S_2+S'_2)\epsilon^2/\gamma^2=0$ for $\epsilon=\epsilon^-_0$.
Therefore, for such large $n$, we have $\epsilon(m_1+m_3)<0$ for $\epsilon=\epsilon^-_0- |\epsilon^-_0|^{\nu}$ and $\epsilon(m_1+m_3)>0$ for $\epsilon=\epsilon^-_0+ |\epsilon^-_0|^{\nu}$.
By the {intermediate} value theorem, this shows that there exists $\epsilon^-\in(\epsilon^-_0- |\epsilon^-_0|^{\nu},\epsilon^-_0+ |\epsilon^-_0|^{\nu})$ such that $m_1+m_3=0$ for $\epsilon=\epsilon^-$.
In view of our previous comments, such $\epsilon^-$ is unique and we must have
\begin{equation*}
	\lambda^-=-\gamma\phi_0 + \epsilon^-.
\end{equation*}
This also establishes
\begin{equation}\label{epsilon-}
	\epsilon^- = -\frac{\sqrt 2}{\sqrt N} (1+o(1)),
\end{equation}
since $|\epsilon^-_0|^{\nu}=o(\epsilon_0^-)$ when $\nu>1$.

Recall the scalar $\alpha=\alpha^-$ for the eigenvalue $\lambda^-$, which is nonzero by Proposition~\ref{important equivalence}.
We assume that $\alpha^->0$ after a rescaling of $\ket{\lambda^-}$.
By Eq.~\eqref{masterequation} and $\sum_{\ell=0}^k \left\| P_{\ell} \ket{\lambda^-} \right\|^2=1$,
we obtain
\begin{equation*}
	1 = (\alpha^-)^2 \left(\frac{4}{N(\epsilon^-)^2}+\frac{2(S_2+S'_2)}{\gamma^2}+O(\epsilon^-)\right).
\end{equation*}
Using Eqs.~{\eqref{two comments}} and \eqref{epsilon-}, we obtain
\begin{equation}\label{alpha-}
	\alpha^- = \frac{1}{2}+o(1).
\end{equation}


Recall that $\ket{\psi(0)}$ is a $\phi_0$-eigenvector of $A$.
By Eq.~\eqref{masterequation} with $\ell=0$, we also obtain
\begin{equation}\label{how s and lambda- interact}
	\braket{\psi(0)}{\lambda^-} = -\frac{ 2\alpha^- }{\sqrt{N}\,\epsilon^-} = \frac{1}{\sqrt{2}}+o(1),
\end{equation}
where we have also used Eqs.~\eqref{epsilon-} and \eqref{alpha-}.

We next find the eigenvalue $\lambda^+$ for the first excited state of $H$.
On the one hand, by the definition of $H$ and $\ket{\psi(0)}$, we have
\begin{equation*}
	\bra{\psi(0)}H\ket{\psi(0)} = -\gamma\phi_0-\frac{2}{N}.
\end{equation*}
On the other hand, we also have
\begin{equation*}
	\bra{\psi(0)}H\ket{\psi(0)} = \lambda^-\big|\braket{\psi(0)}{\lambda^-}\big|^2 + \sum_{\lambda\ne\lambda^-} \lambda \big|\braket{\psi(0)}{\lambda}\big|^2,
\end{equation*}
where the sum on the right-hand side is over the eigenvalues of $H$ on $\mathcal{H}_{\mathrm{inv}}$ other than $\lambda^-$.
By Eqs.~\eqref{epsilon-} and \eqref{how s and lambda- interact}, we have
\begin{equation*}
	\lambda^-\big|\braket{\psi(0)}{\lambda^-}\big|^2 = -\frac{\gamma\phi_0}{2} +o(1),
\end{equation*}
where we also note that $\gamma\phi_0=1+o(1)$ by Eqs.~\eqref{eigenvalues} and \eqref{evaluate gamma}.
Suppose now that $\lambda^+\geqslant -\gamma\phi_1$.
Then, since $\sum_{\lambda}\big|\braket{\psi(0)}{\lambda}\big|^2=1$, it follows from Eqs.~\eqref{denominator} with $\ell=1$ and \eqref{how s and lambda- interact} that
\begin{equation*}
	\sum_{\lambda\ne\lambda^-} \lambda \big|\braket{\psi(0)}{\lambda}\big|^2 \geqslant -\gamma\phi_1 \sum_{\lambda\ne\lambda^-} \big|\braket{\psi(0)}{\lambda}\big|^2 = -\frac{\gamma\phi_1}{2}+o(1) = -\frac{\gamma\phi_0}{2} +\frac{1}{2k}+o(1).
\end{equation*}
However, this implies that
\begin{equation*}
	\bra{\psi(0)}H\ket{\psi(0)} \geqslant -\gamma\phi_0+\frac{1}{2k}+o(1) > -\gamma\phi_0-\frac{2}{N}
\end{equation*}
for large $n$, which is a contradiction.
(Recall that $k$ is fixed.)
Hence, for large $n$, it follows that $\lambda^+ < -\gamma\phi_1$ and thus $\lambda^+\in (-\gamma\phi_0,-\gamma\phi_1)$ by virtue of Corollary~\ref{- gamma d is not eigenvalue}.
Observe that $m_1+m_3$ (again as a function of $\lambda$) is also monotone decreasing on $(-\gamma\phi_0,-\gamma\phi_1)$, from which it follows that $\lambda^+$ is the \emph{unique} eigenvalue of $H$ on $\mathcal{H}_{\mathrm{inv}}$ in this range.
We may now proceed as in the above discussions, and conclude that
\begin{equation*}
	\lambda^+=-\gamma\phi_0+\epsilon^+,
\end{equation*}
where
\begin{equation}\label{epsilon+}
	\epsilon^+ = \frac{\sqrt 2}{\sqrt N} (1+o(1)).
\end{equation}
We also obtain the scalar $\alpha=\alpha^+$ for the eigenvalue $\lambda^+$ as

\begin{equation}\label{alpha+}
	\alpha^+ = \frac{1}{2}+o(1),
\end{equation}
and
\begin{equation}\label{how s and lambda+ interact}
	\braket{\psi(0)}{\lambda^+} = -\frac{ 2\alpha^+ }{\sqrt{N}\,\epsilon^+} = -\frac{1}{\sqrt{2}}+o(1).
\end{equation}


Earlier, we referenced a ``bad'' eigenvalue arising from the condition $m_1-m_3=0$, which we have effectively excluded by limiting the dynamics to $\mathcal{H}_{\mathrm{inv}}$. As a function of $\lambda=-\gamma\phi_0+\epsilon$, the expression $m_1-m_3$ demonstrates a monotone decrease over the interval $(-\infty,-\gamma\phi_1)$. Within this range, there exists a unique $\lambda$ for which $m_1-m_3=0$. Following a similar approach to that in Eq.~\eqref{m1+m3}, we derive the following equation:
\begin{equation*}
    m_1-m_3=1-\frac{S_1-S_1'}{\gamma}-\frac{S_2-S_2'}{\gamma^2}\epsilon+O(\epsilon^2)=\frac{2S_1'}{\gamma}-\frac{S_2-S_2'}{\gamma^2}\epsilon+O(\epsilon^2),
\end{equation*}
where we have applied Eq.~\eqref{gamma}. Utilizing Eqs.~\eqref{S}, \eqref{S'}, and \eqref{evaluate gamma}, we can approximate the corresponding $\epsilon$ as:
\begin{equation*}
    \epsilon=\frac{2S_1'\gamma}{S_2-S_2'}=O\!\left(\frac{1}{n^{\delta}}\right).
\end{equation*}
Considering Eqs.~\eqref{two comments}, \eqref{epsilon-}, and \eqref{epsilon+}, this eigenvalue $\lambda$ exhibits particularly problematic behavior when $\delta>k/2$, in the sense that $\lambda\in(\lambda^-,\lambda^+)$ for large $n$. Nevertheless, this adverse behavior does not affect the dynamics since $\braket{\psi(0)}{\lambda}=0$.

\subsection{Computational complexity}

Now we apply the general recipe outlined in the previous section to the search algorithm on $J(n,k)$ by fixing $k$ and letting $n\rightarrow\infty$. Using Eqs.~\eqref{alpha-} and~\eqref{alpha+}, we check that $\braket{w}{\lambda^+} = \braket{w}{\lambda^-} + o(1)$. Using~\eqref{how s and lambda- interact}  and~\eqref{how s and lambda+ interact}, we check that $\braket{\lambda^+}{\psi(0)} = -\braket{\lambda^-}{\psi(0)} + o(1).$
Thus, condition~\eqref{bra_lam_w} is true.

The probability of finding a marked vertex as a function of time is given by Eq.~\eqref{eq:p(t)=sin^2} and then using Eq.~\eqref{final_t_run} the optimal running time is
\begin{equation*}
	t_{\mathrm{run}} = \frac{\pi\sqrt{N}}{2\sqrt 2},
\end{equation*}
and using Eq.~\eqref{final_p_succ} the success probability is 
\begin{equation}\label{success Johnson}
	p_{\mathrm{succ}}=1+o(1).
\end{equation}
Note that the asymptotic success probability does not depend on $\delta$, the distance between the two marked vertices $w_1$ and $w_2$.

We have so far assumed that $\delta$ is known in advance; see Eq.~\eqref{gamma}.
In the case where we have no prior information on $\delta$, we simply apply the above search algorithm in turn for $\delta=1,2,\dots,k$.
We note that we can check if the outcome is a marked vertex by applying the oracle for each of the $k$ steps.
Hence, the success probability is still given by Eq.~\eqref{success Johnson}.
The running time is at most
\begin{equation*}
	\frac{\pi k\sqrt{N}}{2\sqrt 2}.
\end{equation*}
We fix $k$, so this is still of order $O(\sqrt{N})$.

\section{Final remarks}\label{sec:conc}

We have developed a framework to assess analytically the computational complexity of spatial search algorithms using continuous-time quantum walks on graphs with multiple marked vertices. This framework is underpinned by four propositions that not only lay the mathematical foundation but also offer a practical method for computing the eigenvalues and eigenvectors of the Hamiltonian. Additionally, these propositions facilitate the determination of the optimal value for the parameter $\gamma$ in the search algorithm. A key step in this process involves calculating a basis of eigenvectors of the kernel of a symmetric matrix, the dimension of which equals the number of marked vertices. This method applies to any graph whose associated Hamiltonian exhibits a sufficiently small gap between $\lambda^\pm$. In simpler scenarios, the success probability of the search algorithm can be represented as $a^2\sin^2\epsilon t$, with $2\epsilon$ usually being the gap. The optimal running time is $t=\pi/2\epsilon$, leading to a success probability $p_{\mathrm{succ}}=a^2$. For the search algorithm to achieve optimality, implying a time complexity of $O(\sqrt{N})$, it is crucial that $\epsilon$ scales as $O(1/\sqrt{N})$ and $a$ maintains an order of $\Omega(1)$.

We demonstrated our framework with an analysis of the spatial search algorithm on the Johnson graph $J(n,k)$ featuring two marked vertices. Our analysis indicates that, with $k$ fixed and the distance $\delta$ between the marked vertices known beforehand, the optimal running time is $\pi\sqrt{N}/(2\sqrt{2})$, where $N$ is the total number of vertices. In cases where $\delta$ is not known, the algorithm must be executed $k$ times, each assuming a different $\delta$. This results in a total running time of at most $\pi k\sqrt{N}/(2\sqrt{2})$, maintaining a success probability of $1+o(1)$ as $n$ approaches infinity.

Usually, the condition for the search algorithm be optimal is that the spectral gap of the Hamiltonian is sufficiently larger than the overlap of the initial condition with the
marked vertex~\cite{CNR20b}. This condition holds true in the cases analyzed in the literature when there is only one marked vertex, to the best of our knowledge. The problem becomes more intricate in the presence of multiple marked vertices since the time complexity depends on specific vertex placements in the general case. As discussed above, in Johnson graphs with two marked vertices, the time complexity is influenced by the distance between the marked vertices.

The results of this work can accelerate numerical calculations of the computational complexity of search algorithms on arbitrary graphs. In this context, certain analytical expressions from Sections~\ref{asymptotic} and~\ref{subssec:compcomplexity} can be circumvented. The numerical approach is outlined as follows. Start by determining numerically $\phi_0$ and $\ket{\phi_0}$. Next, use the entries of matrix $M^\lambda$, as defined in Eq.~\eqref{eq:Mww}, to identify $\lambda^\pm$ and $\gamma$. Since they need to be found concurrently, we assign numerical values to both $\gamma$ and $\lambda$ and then check whether det$(M^\lambda)=0$. If this is not the case, we change $\gamma$ and $\lambda$ until locating the roots. Following this, $\epsilon$ is deduced from $|\lambda^-+\gamma\phi_0|$. We then we proceed to numerically calculate $\braket{\lambda^\pm}{\phi_0}$ and $\braket{w}{\lambda^\pm}$ for all $w$. $\braket{\lambda^+}{\phi_0}$ and $\braket{w}{\lambda^+}$ must be nonzero for at least one $w$, otherwise $\lambda^+$ does not fulfill the required conditions. Afterwards, we compute the success probability $p(t)$ using Eq.~\eqref{eq:p(t)-2terms}. By plotting $p(t)$, we can evaluate whether the curve resembles the square of a sinusoidal function and whether the peak is prominent.

\section*{Acknowledgements}
The authors thank the editors of ACM Transactions on Quantum Computing for the invaluable help in improving this paper.
The work of P.H.G. Lugão was supported by CNPq grant number 140897/2020-8.
The work of R. Portugal was supported by FAPERJ grant number CNE E-26/202.872/2018, and CNPq grant number 308923/2019-7. 
The work of M. Sabri was supported by JST SPRING grant number JPMJSP2114.
The work of H. Tanaka was supported by JSPS KAKENHI grant numbers JP20K03551 and JP23K03064.


\appendix
\section*{Appendix}
\section{The proof of Eq.~\eqref{S'}}
\label{sec: proofs of eqs}

We first invoke some results from~\cite{Terwilliger2005DCC} to evaluate the $_3F_2$ series in Eq.~\eqref{hypergeometric}.
Set
\begin{equation*}
    \phi_{\ell}=\phi_0+h\ell (\ell+1+s), \qquad \phi^*_{\ell}=\phi_0^*+s^*\ell
\end{equation*}
for $0\leqslant\ell\leqslant k$, and
\begin{equation*}
    \alpha_{\ell}=hs^*\ell(\ell-k-1)(\ell+r), \qquad \beta_{\ell}=hs^*\ell(\ell-k-1)(\ell+r-s-k-1)
\end{equation*}
for $1\leqslant\ell\leqslant k$, where
\begin{equation*}
    s=-n-2, \qquad r=k-n-1,
\end{equation*}
and $\phi_0,\phi_0^*,h$, and $s^*$ are arbitrary with $h,s^*\ne 0$.
The sequence
\begin{equation*}
    \big(\phi_{\ell},\phi_{\ell}^*,\ell=0,\dots,k;\alpha_{\ell'},\beta_{\ell'},\ell'=1,\dots,k\big)
\end{equation*}
is a \emph{parameter array} of dual Hahn type; cf.~\cite[Example 5.12]{Terwilliger2005DCC}.
Note that the $\phi_{\ell}$ agree with the eigenvalues of $J(n,k)$ (cf.~Eq.~\eqref{eigenvalues}) provided that we set $\phi_0=k(n-k)$ and $h=1$.
Moreover, we have
\begin{equation*}
    _3F_2\!\left.\left(\begin{matrix} -\ell,-\delta,\ell-n-1 \\ k-n,-k\end{matrix}\,\right|1\right) = f_{\delta}(\phi_{\ell}),
\end{equation*}
where
\begin{equation*}
    f_{\delta}(x) = \sum_{\nu=0}^{\delta} \frac{(x-\phi_0)\cdots(x-\phi_{\nu-1})(\phi_{\delta}^*-\phi_0^*)\cdots(\phi_{\delta}^*-\phi_{\nu-1}^*)}{\alpha_1\cdots\alpha_{\nu}}.
\end{equation*}
By~\cite[Theorem 4.1, Lemma 4.2]{Terwilliger2005DCC}, we have
\begin{equation*}
    f_{\delta}(x) = \frac{\beta_1\cdots\beta_{\delta}}{\alpha_1\cdots\alpha_{\delta}} f_{\delta}^{\Downarrow}(x) = \frac{\delta!}{(k-n)_{\delta}} f_{\delta}^{\Downarrow}(x),
\end{equation*}
where
\begin{equation*}
    f_{\delta}^{\Downarrow}(x) = \sum_{\nu=0}^{\delta} \frac{(x-\phi_k)\cdots(x-\phi_{k-\nu+1})(\phi_{\delta}^*-\phi_0^*)\cdots(\phi_{\delta}^*-\phi_{\nu-1}^*)}{\beta_1\cdots\beta_{\nu}},
\end{equation*}
and $(a)_m=a(a+1)\cdots(a+m-1)$ denotes the shifted factorial.
We have
\begin{equation*}
    f_{\delta}^{\Downarrow}(\phi_{\ell}) = {}_3F_2\!\left.\left(\begin{matrix} \ell-k,-\delta,n-k-\ell+1 \\ 1,-k\end{matrix}\,\right|1\right),
\end{equation*}
and hence
\begin{equation}\label{modified 3F2}
    _3F_2\!\left.\left(\begin{matrix} -\ell,-\delta,\ell-n-1 \\ k-n,-k\end{matrix}\,\right|1\right) = \frac{\delta!}{(k-n)_{\delta}} {}_3F_2\!\left.\left(\begin{matrix} \ell-k,-\delta,n-k-\ell+1 \\ 1,-k\end{matrix}\,\right|1\right).
\end{equation}

By Eqs.~\eqref{eigenvalues}, \eqref{Pw}, \eqref{hypergeometric}, and \eqref{modified 3F2}, we have
\begin{equation*}
	\frac{\bra{w_1}P_{\ell}\ket{w_2}}{\phi_0-\phi_{\ell}}=O\!\left(\frac{1}{n^{k-\ell+1+\delta-\tau}}\right), \qquad \frac{\bra{w_1}P_{\ell}\ket{w_2}}{(\phi_0-\phi_{\ell})^2}=O\!\left(\frac{1}{n^{k-\ell+2+\delta-\tau}}\right)
\end{equation*}
for $1\leqslant \ell\leqslant k$, where $\tau=\min\{k-\ell,\delta\}$, and Eq.~\eqref{S'} follows.

We remark that the above method can be generalized to other families of $Q$-polynomial distance-regular graphs, again including the Hamming graphs (see~\cite[Chap.~9]{BCN1989B}).


\end{document}